%% file: main.tex
\documentclass[final,5p,times,twocolumn]{elsarticle}

\usepackage{amssymb}
\usepackage{amsmath}
\usepackage{amsthm}
\usepackage{hyperref}
\usepackage{caption}
\usepackage{subcaption}
\usepackage{color}
\usepackage[table]{xcolor}
\usepackage{tabularx}

\theoremstyle{definition}

\newtheorem{remark}{Remark}[section]
\newtheorem{proposition}{Proposition}[section]

\makeatletter
\def\ps@pprintTitle{%
  \let\@oddhead\@empty
  \let\@evenhead\@empty
  \def\@oddfoot{\hfill\thepage\hfill}%
  \let\@evenfoot\@oddfoot
}
\makeatother

\begin{document}

\begin{frontmatter}



\title{Measures of classification bias derived from sample size analysis}


\author[DU]{Ioannis Ivrissimtzis} 
\author[DU]{Shauna Concannon} 
\author[SL]{Matthew Houliston} 
\author[SL]{Graham Roberts} 

\affiliation[DU]{organization={University of Durham}, country={UK}}
\affiliation[SL]{organization={Servelegal Ltd}, country={UK}} 

\begin{abstract}
We propose the use of a simple intuitive principle for measuring algorithmic classification bias: the significance of the differences in a classifier's error rates across the various demographics is inversely commensurate with the sample size required to statistically detect them. That is, if large sample sizes are required to statistically establish biased behavior, the algorithm is less biased, and vice versa. In a simple setting, we assume two distinct demographics, and non-parametric estimates of the error rates on them, $\varepsilon_1$ and $\varepsilon_2$, respectively. We use a well-known approximate formula for the sample size of the chi-squared test, and verify some basic desirable properties of the proposed measure. Next, we compare the proposed measure with two other commonly used statistics, the difference $\varepsilon_2 - \varepsilon_1$ and the ratio $\varepsilon_2 / \varepsilon_1$ of the error rates. We establish that the proposed measure is essentially different in that it can rank algorithms for bias differently, and we discuss some of its advantages over the other two measures. Finally, we briefly discuss how some of the desirable properties of the proposed measure emanate from fundamental characteristics of the method, rather than the approximate sample size formula we used, and thus, are expected to hold in more complex settings with more than two demographics. 
\end{abstract}







\end{frontmatter}


\input{sec1}
\input{sec2}

\input{sec3}

\input{sec4}
\input{sec5}

\input{sec6}
\input{sec7}


\bibliographystyle{elsarticle-num}
\bibliography{ref}


\end{document}

%% file: sec1.tex
\section{Introduction} 
\label{sec:sec1}

We study the behavior of classifiers tested on datasets stratified in several classes, called {\em demographics}. In a typical application, the elements of the test dataset would correspond to humans, while the classes could correspond to characteristics such as race, gender, or age. 

We assume that we have non-parametric point estimates of the error rates of the classifier on each demographic. Typically, these would be empirical estimates obtained by running the algorithm on the test set and counting the rate of unsuccessful classifications on each demographic. Our aim is to propose and study a metric that will quantify the {\em bias} of the classifier, that is, measure the significance of differences in algorithmic performance across demographics. 

The main idea, which is intuitive, simple, and general, can be formulated as the following principle:
\begin{quote}
The bias of an algorithm is inversely commensurable to the sample size $N$ of a statistical test of given {\em power} to detect it with a certain {\em confidence}. 
\end{quote}
For example, if a dataset consists of two demographics, $D_1$ and $D_2$, and the error rates of an algorithm on them are $\varepsilon_1$ and $\varepsilon_2$, respectively, to answer the question of how significant the difference between $\varepsilon_1$ and $\varepsilon_2$ is, we look at how many samples from the dataset are required to statistically detect that difference. The more samples we need, the smaller the bias of the classifier. 

Notice that here we do not assess the statistical significance of the point estimates of the algorithm's error rates on each demographic. We take them as given, and we want to arrive at an assessment of the bias of the algorithm given these point estimates. Although this might seem as a limitation, as shifting the problem rather than addressing it, we believe that, nevertheless, it is an essential property for establishing a coherent bias metric. It keeps separate two issues often conflated in the literature; what we measure and the process of measuring it. Indeed, if instead we had defined bias as a measure of the probability that the error rates $\varepsilon_1$ and $\varepsilon_2$ are different, then the measured bias would most likely increase as we test the algorithm more. However, that increase in measured bias would not reflect any change in the behavior of the algorithm, but rather a decrease of our uncertainty over its behavior. 


\subsection{Paper organization}

In Section~\ref{sec:sec2}, we discuss the background and related work. In Section~\ref{sec:sec3}, we apply the proposed general principle on the specific instance of classifiers on datasets consisting of two demographics. We adopt the chi-squared test as the means of statistically detecting performance differences and use a well-established approximate closed-form formula to estimate sample sizes. We show that the sample size depends on the power and confidence of the test by just a multiplicative constant and establish some other properties that are desirable for any measure of bias. 

In Section~\ref{sec:sec4}, we compare the behavior of the proposed bias measure against two widely used metrics, the difference $\varepsilon_2 - \varepsilon_1$ and the ratio $\varepsilon_2 / \varepsilon_2$ of the error ratios. First, we study the behavior of the proposed metric when the difference or ratio of the error rates is constant. Then we illustrate the method with actual data from a machine learning classification problem of face liveness. 

In Section~\ref{sec:sec5} we further justify our approach arguing that, more generally, for any number of demographics, any statistical test with some basic desirable properties will lead to a measure of bias with some equivalent desirable properties. For example, when the error rates on all demographics tend to zero, the bias would also tend to zero. 

In Section~\ref{sec:sec6}, we discuss as potential directions for future research some alternatives to the standard approach to sample size analysis adopted here, and we briefly conclude in Section~\ref{sec:sec7}.


\subsection{Contributions and limitations}

The main contribution of the paper is a new algorithmic bias metric based on sample size analysis. With the caveat of a limitation discussed below, the proposed metric has the characteristics of a good metric put forward in \cite{hand2009} and further articulated in \cite{anagnostopoulos2017}. It captures coherently the aspect of performance of interest, sharply separating the algorithmic bias from the uncertainty in its measurement. It is simple to report as a single number, intuitive and, at least in the simple setting of two demographics, computationally tractable. 

As a secondary contribution, in Section \ref{sec:sec4-2}, in the comparison of the proposed metric with the ratios of error rates, we highlight an aspect of algorithmic bias that we believe it has been largely overlooked in the literature. As algorithms evolve and we move towards a zero tolerance of misclassifications, the behavior of a metric at the vicinity of the zero error rate becomes important and should be studied. In particular, a bias metric should be relatively stable with respect to the direction from which we approach the point of overall zero error rate in the space of the error rates on demographics. 

The main limitation of the proposed bias metric is its dependence on two continuous variables, the power and significance of the statistical test, and one categorical variable, the type of statistical test. Although there are some values of these continuous variables that are widely regarded in the literature as standard, for example a power of 0.80 and a confidence of 0.95, here, for a more transparent approach, we do not fix them, using the proposed metric just for ranking algorithms for bias, rather than as an absolute measure. It is a more complicated question, out of the scope of the current paper, to study the implications of choosing a statistical test other than the chi-squared.

In a secondary limitation, here we refrain from proposing an explicit function of bias as an inversely commensurable quantity of the sample size $N$. Although there are some obvious candidates, such as $1/N$, or $-\log N$, for a more transparent approach, not having a clear justification for a choice, we preferred to work directly with $N$ as a measure of fairness that is inversely commensurable with bias. However, we note that $N$ is not just intuitive, but also has a clear physical meaning, telling us how much an algorithm can be used before bias becomes statistically detectable.




%% file: sec2.tex
\section{Background}
\label{sec:sec2}

The study of algorithmic bias has emerged as an essentially interdisciplinary field, attracting researchers from diverse areas, such as computer science, mathematics, social sciences, and also law. 

\cite{barocas2016} was an influential paper that shaped the legal perspective on problem. Several recent surveys and scoping reviews focus on various other aspects of the problem. \cite{das2023} gives a concise review of algorithmic bias metrics, with credit scoring as its target application. The overarching \cite{schwartz2022} reviews the plethora of proposed approaches to algorithmic bias one has to navigate. \cite{pagano2023}, based on the PRISMA guidelines, uses several bias metrics and datasets in a comprehensive technical comparative study. 

In \cite{corbett2023}, the emphasis is on the alignment of the adopted definitions of fairness with the stated policy goals they serve, arguing for more context-specific rather that axiomatic definitions of algorithmic bias. \cite{barr2024} proposes a practical framework for assisting the selection of context appropriate bias metrics. \cite{friedler2021}, in the opposite direction, studies the assumptions principles and beliefs, that, implicitly most of the times, underline the proposed definitions of bias. 

We note that despite the considerable research effort in understanding the ethical, legal, policy, and regulatory implications of the various measures of bias, most of the proposed metrics are tailored towards the needs of the researchers. Their aim is to obtain insights that will support corrective interventions, either during training or inference, through an essentially white-box analysis. In \cite{ivrissimtzis2024}, zero-failure testing is proposed as a black-box approach to algorithmic performance assessment that could be more suitable for policy-makers and regulators.  


\subsection{Measuring bias}

Similarly to our approach, several proposed metrics of bias aim at assessing the fairness of an algorithm with respect to various groups of people. In the seminal \cite{hardt2016}, the Equalized Odds and the Equal Opportunity were introduced as fairness measures with respect to the accuracy rate or the TPR of an algorithm. \cite{martinez2020} proposed minimax Pareto fairness, defining group fairness through the worst case scenario. The method was further studied and evaluated in \cite{diana2021}. \cite{dealcala2023} proposed the use of the difference of accuracy rates normalized by the empirical standard deviation as a bias metric. \cite{suarez2025} proposed a framework for the use of classification bias metrics in regression problems with continuous outcomes. 


In a different approach, several metrics have been proposed aiming to establish fairness with respect to the individual. This approach was pioneered in \cite{dwork2012}, where the concept of a task-specific similarity metric for the individual was introduced. A significant strand of research on individual fairness is based on the concept of counterfactual fairness, introduced in \cite{kusner2017}, employing a form of causation analysis through Structural Causal Models \cite{pearl2009}. In \cite{mishler2021}, they propose a post-processed predictor for achieving counterfactual fairness in terms of equalized odds. In \cite{zhou2024}, an algorithm for a controlled trade-off between counterfactual fairness and algorithmic performance is proposed and validated.


\subsubsection{Intersectional bias}

Recent research, increasingly focuses on the problem of intersectional bias, where not only we might have more than two groups arising from one attribute, but overlapping groupings arising from more than one attributes. The relevance of the problem was first demonstrated in \cite{buolamwini2018}. In \cite{foulds2020}, a metric for intersectional fairness based on the 80\% rule was proposed. \cite{ghosh2021} uses minmax to define a worst case group fairness, this time in an intersectional setting. \cite{wastvedt2024} extends the work in \cite{mishler2021}, introducing counterfactual equalized odds in an intersectional setting. \cite{shukla2025} introduces intersectional sensitivity, a metric based on Wasserstein distances between ideal distributions, and uses it for the systematic analysis of the intersectional interactions of bias.

Here, we do study the intersectional bias problem, even though, in principle, the proposed sample size analysis approach can directly be adapted to it.


\subsection{Analyzing and correcting algorithmic bias}

Given the implications of the existence of algorithmic bias in deployed systems, there is a large body of research going beyond the fundamental problem of measuring it, and focusing on identifying its causes and various manifestations, as well as on proposing ways to fix it. 

\cite{amini2019}, uses a variational autoencoder to learn a latent distribution within the dataset, and adjusts the training process, increasing the weight of unrepresented regions. The measure of bias they adopt in their validating experiments is the variance in the precision of the algorithm over the various demographics. \cite{karkkainen2021} constructs a large, balanced on race, face image dataset, and reports improved and consistent across race and gender groups performance of models trained on it. \cite{delobelle2022} studies the significantly more challenging problem of biased language, comparing several pretrained large language models on various metrics, suggesting, as a more practical approach, the evaluation of language bias on downstream tasks. \cite{krug2025} studies bias on large models pretrained on ImageNet, which often serve as backbones for downstream computer vision tasks. Finally, \cite{kumar2025} proposes a method for reducing biases in computer vision models through adjustments to the data augmentation process.






%% file: sec3.tex
\section{A bias measure for binary classifiers} 
\label{sec:sec3}

Let $C$ be a binary classifier, and $D = \{ D_1 \cup D_2 \}$ a test dataset that is the union of two disjoint demographics, $D_1$ and $D_2$. Our aim is to define and study a measure of the bias of $C$ over the two demographics, implementing the general principle discussed in Section~\ref{sec:sec1}. 

\smallskip 

Notice that in our setting, the algorithm $C$ does not classify a sample in $D$ as an element of demographic $D_1$ or demographic $D_2$. Instead, we have an unspecified binary classification task performed on the elements of $D$, with error rates $\varepsilon_1$ and $\varepsilon_2$ on $D_1$ and $D_2$, respectively. For simplicity, we may also assume that $D$ consists of samples from one class only, say positive samples, in which case $\varepsilon_1$ and $\varepsilon_2$ are the False Negative Rates of $C$ on $D_1$ and $D_2$, respectively. However, this simplifying assumption is not necessary for the validity of our study. 

\smallskip 

With bias understood here as a difference in the classifier's performance over demographics $D_1$ and $D_2$, the null hypothesis on an unbiased $C$ is
$$
H_0 : \varepsilon_1 = \varepsilon_2
$$
The sample size analysis will be done for the $\chi$-squared test, which in this setting is the simplest and most natural choice. 

\smallskip 


There are various practical approaches to the sample size computation, each depending on different approximation assumptions, as for example the approximation of the discrete binomial distributions of $\varepsilon_1$ and $\varepsilon_2$ by a continuous normal distribution. In \cite{ryan2013}, a formula for the sample size $N$ that is often used by practitioners is given as 
\begin{equation}
N_{\alpha,\beta} (\varepsilon_1,\varepsilon_2) = \displaystyle{\frac{1}{2}} \left( \displaystyle{\frac{z_\beta + z_\alpha}{\arcsin(\sqrt{\varepsilon_1})-\arcsin(\sqrt{\varepsilon_2})}} \right)^2
\label{eq:twoproportions-arcSin}
\end{equation} 
where $\alpha$ is the level of significance, $1-\beta$ is the power of the test, and $z_x$ denotes the critical value of the normal distribution at $x$. The $\arcsin$ transformation of the proportions $\varepsilon_1$ and $\varepsilon_2$ in Eq.~\ref{eq:twoproportions-arcSin} is a standard statistical technique justified when the values of either of the proportions $\varepsilon_1, \varepsilon_2$ deviate significantly from 0.5. Notice that in one-sided tests, instead of $\alpha$ one has to use $\alpha / 2$. Also, in the literature, the accuracy rates $1-\varepsilon_1, 1-\varepsilon_2$ are often used in place of the error rates $\varepsilon_1, \varepsilon_2$, but the two expositions are equivalent up to a swap of class labels. 

Proposition~\ref{prop:scalingProperty} states that the proposed measure of binary classification bias is invariant with respect to the values of $\alpha$ and $\beta$, up to a multiplicative constant. 

\begin{proposition}
\label{prop:scalingProperty}
For any $0 \leq \varepsilon_1 \leq \varepsilon_2 \leq 1$, we have 
$$
\displaystyle{\frac{N_{\alpha_1,\beta_1} (\varepsilon_1,\varepsilon_2)}{N_{\alpha_2,\beta_2} (\varepsilon_1,\varepsilon_2)}} = c
$$
where $c$ is constant that depends on $\alpha_1,\alpha_2,\beta_1,\beta_2$ only. 
\end{proposition}

\begin{proof}
From Eq.~\ref{eq:twoproportions-arcSin} we have 
$$
\displaystyle{\frac{N_{\alpha_1,\beta_1} (\varepsilon_1,\varepsilon_2)}{N_{\alpha_2,\beta_2} (\varepsilon_1,\varepsilon_2)}} = 
\displaystyle{\left( \frac{z_{\alpha_1} + z_{\beta_1}}{z_{\alpha_2} + z_{\beta_2}} \right)^2}.
$$
\end{proof}

\begin{remark}
Proposition~\ref{prop:scalingProperty} allows us to study the sample size $N_{\alpha,\beta} (\varepsilon_1,\varepsilon_2)$ as estimated by Eq.~\ref{eq:twoproportions-arcSin} for fixed values of $\alpha$ and $\beta$, knowing that any other values of these parameters would only scale the sample size estimates by a constant. Effectively, it allows us to marginalize the variables $\alpha$ and $\beta$, which do not have an obvious natural relation to the concept of bias, and obtain results that depend on $\varepsilon_1$ and $\varepsilon_2$ only. In particular, if for a given pair of values of $\alpha$ and $\beta$ a classifier $C_1$ is deemed less biased on $D$ than a classifier $C_2$, then the same holds for any other values of $\alpha$ and $\beta$. 
\end{remark}

\begin{remark}
We note that scale invariance is a property that many widely used metrics, starting from those based on the Euclidean distance such as the Mean Square Error, do not possess, and the issue is usually addressed through a data normalization whose type depends on the context of the application. Similarly, the proposed measure of bias is not invariant with respect to the values of $\alpha$ and $\beta$. That means that for the following three possible applications of a bias measure: 
\begin{enumerate}[(i)]
    \item Quantify the bias of a classifier $C$ as a function of the error rates on the two demographics $D_1$ and $D_2$. 
    \item Quantify the difference in bias of two classifiers $C_1$ and $C_2$ as a function of their error rates on $D_1$ and $D_2$. 
    \item Rank two classifiers $C_1$ and $C_2$ according to their bias on $D_1$ and $D_2$. 
\end{enumerate}
for (i) and (ii) we obtain partial results, that is, up to a multiplicative constant, and a complete result for (iii). 
\end{remark}

For the rest of the paper we will fix the level of significance at the default value of $\alpha = .95$ and the value of the power of the test at $1-\beta = .90$. These are two commonly used values of the two statistical quantities, even though arbitrary for our purposes. Fig.~\ref{fig:colourmap-p1-p2} shows a colormap of the natural logarithm of the sample size function $N_{\alpha,\beta} (\varepsilon_1,\varepsilon_2)$ in Eq.~\ref{eq:twoproportions-arcSin}, for the default values of $\alpha=.95$ and $\beta=.90$. The logarithmic transformation was applied to compress the range and improve visualization. 
\begin{figure}[ht]
    \centering
    \includegraphics[width=0.85\columnwidth]{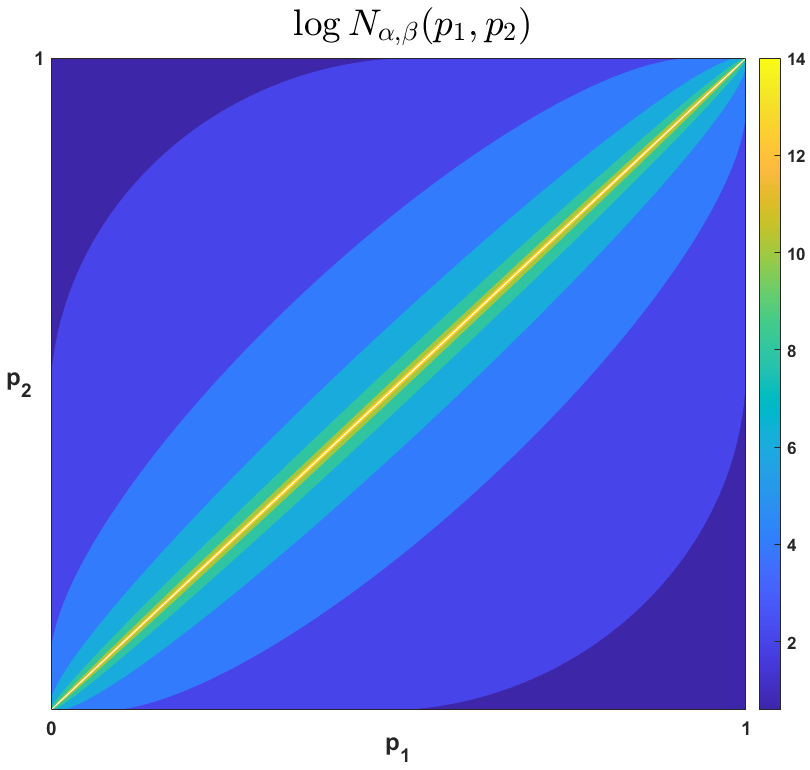}
    \caption{A colormap of the values of $\log N_{\alpha,\beta} (\varepsilon_1,\varepsilon_2)$.}
    \label{fig:colourmap-p1-p2}
\end{figure}

From Fig.~\ref{fig:colourmap-p1-p2}, we notice that the highest values of $N_{\alpha,\beta} (\varepsilon_1,\varepsilon_2)$ are near the diagonal $\varepsilon_1 = \varepsilon_2$. In fact, Proposition~\ref{prop:noBias} states that as $\varepsilon_2 \rightarrow \varepsilon_1$ the value of $N$ goes to infinity, and thus, its inverse $1/N$ goes to 0. This is a desirable property for any measure of bias, because if a classifier has equal error rates over two demographics, then there is no bias over them. 

\begin{proposition}
\label{prop:noBias}
For any $0 < \varepsilon_1 < 1$, we have
$$
\lim_{\varepsilon \rightarrow \varepsilon_1} N_{\alpha_1,\beta_1} (\varepsilon, \varepsilon_1) = \infty.
$$
\end{proposition}

\begin{proof}
Since square root and arcsin are continuous in (0,1), we have 
$$
\lim_{\varepsilon \rightarrow \varepsilon_1} (\arcsin(\sqrt{\varepsilon})-\arcsin(\sqrt{\varepsilon_1})) = 0 
$$
giving, 
$$
\lim_{\varepsilon \rightarrow \varepsilon_1} \displaystyle{\frac{1}{2}} \left( \displaystyle{\frac{z_\beta + z_\alpha}{\arcsin(\sqrt{\varepsilon})-\arcsin(\sqrt{\varepsilon_1})}} \right)^2 = \infty
$$
\end{proof}



Proposition~\ref{prop:biasMonotonicity} states that if the error on the demographic where the algorithm is more accurate is kept constant and the error on the other demographic increases, then the bias increases. This is another desirable property that any reasonable bias metric should have.

\begin{proposition}
\label{prop:biasMonotonicity}
For every $0 \leq \varepsilon_1 \leq \varepsilon_2 \leq \varepsilon_3 \leq 1$ 
\begin{equation}
N_{\alpha,\beta} (\varepsilon_1, \varepsilon_3) \leq N_{\alpha,\beta} (\varepsilon_1, \varepsilon_2)    
\end{equation}
\end{proposition}

\begin{proof}
From the formula for $N_{\alpha,\beta}$ in Eq.~\ref{eq:twoproportions-arcSin}, it suffices to show 
$$
\arcsin(\sqrt{\varepsilon_1})-\arcsin(\sqrt{\varepsilon_2}) \ \leq\ 
\arcsin(\sqrt{\varepsilon_1})-\arcsin(\sqrt{\varepsilon_3})
$$
or, 
$$
\arcsin(\sqrt{\varepsilon_3}) \leq \arcsin(\sqrt{\varepsilon_2})
$$
which holds because $\arcsin(x)$ and thus $\arcsin(\sqrt{x})$ increase monotonically in [0,1].     
\end{proof}

%% file: sec4.tex
\section{Comparison with with ratios and differences} 
\label{sec:sec4}

In this section, we compare the proposed bias metric with two other natural bias metrics. The first is the difference of the error rates on the two demographics 
\begin{equation}
D_{err}(\varepsilon_1, \varepsilon_2) = \varepsilon_2 - \varepsilon_1    
\label{eq:difference}
\end{equation}
and the second is the ratio of the error rates 
\begin{equation}
R_{err}(\varepsilon_1, \varepsilon_2) = \displaystyle{\frac{\varepsilon_2}{\varepsilon_1}}
\label{eq:errorRateRatio}
\end{equation}
For simplicity of exposition, we do not consider the direction of bias and assume that $\varepsilon_1 \leq \varepsilon_2$, that is, the error of the classifier is lower on demographic $D_1$. 


\subsection{Ranking comparisons}
\label{sec:sec3-1}

First, we show that the three approaches to measuring bias are essentially different, in the sense that they can yield different bias rankings of classification algorithms. We use the data reported in \cite{wu2023}. There, five in total classifiers for the face recognition problem, corresponding to different combinations of loss, model, and training set, are tested on a set comprising four races (Asian, Black, Indian, White) and two genders (Female, Male). TPRs on each demographic are reported, as well as the differences of the TPRs between Female and Male, for each race separately. 

Table~\ref{tbl:realExample} shows the TPRs of each algorithm on each demographic as reported in \cite{wu2023}; the differences of the TPRs employed in that paper as a measure of bias, and for comparison the sample sizes computed from \ref{eq:twoproportions-arcSin}. The different rankings produced by the two bias metrics are highlighted with a color scheme. 

\begin{table*}[ht]
\begin{center}
\begin{tabular}{c||c|c||c|c|c||c|c|c} 
\rule{0pt}{1em} \cellcolor{black!18}{\bf Asian} & AF & AM & $D_{err}$ & $R_{err}$ & $N$ & $D_{err}\downarrow$ & $R_{err}\downarrow$\ & \ \ $N\uparrow$\ \ \\ \hline
alg.1 & 65.00 & 79.78 & 14.78 & 1.73 & 154 & \cellcolor{red!50}\bf{5} & \cellcolor{green!50}\bf{1} & \cellcolor{green!20}\bf{2} \\ \hline 
alg.2 & 81.56 & 94.44 & 12.89 & 3.32 & 101 & \cellcolor{red!20}\bf{4} & \cellcolor{red!50}\bf{5} & \cellcolor{red!50}\bf{5} \\ \hline 
alg.3 & 81.56 & 93.11 & 11.56 & 2.68 & 135 & \cellcolor{green!20}\bf{2} & \cellcolor{green!20}\bf{2} & \cellcolor{yellow!35}{\bf 3} \\ \hline 
alg.4 & 81.22 & 93.00 & 11.78 & 2.68 & 131 & \cellcolor{yellow!35}{\bf 3} & \cellcolor{yellow!35}{\bf 3} & \cellcolor{red!20}\bf{4} \\ \hline 
alg.5 & 90.00 & 96.78 &  6.78 & 3.11 & 214 & \cellcolor{green!50}\bf{1} & \cellcolor{red!20}\bf{4} & \cellcolor{green!50}\bf{1} \\ 
\end{tabular} 
\vskip 12pt
\end{center} 
\begin{center}
\begin{tabular}{c||c|c||c|c|c||c|c|c} 
\rule{0pt}{1em} \cellcolor{black!18}{\bf Black} & BF & BM & $D_{err}$ & $R_{err}$ & $N$ & $D_{err}\downarrow$ & $R_{err}\downarrow$\ & \ \ $N\uparrow$\ \ \\ \hline
alg.1 & 85.56 & 86.78 & 1.22 & 1.09 & 13708 & \cellcolor{yellow!35}{\bf 3} & \cellcolor{green!20}\bf{2} & \cellcolor{yellow!35}{\bf 3} \\ \hline 
alg.2 & 91.00 & 94.22 & 3.22 & 1.56 & 1118 & \cellcolor{red!50}\bf{5} & \cellcolor{red!50}\bf{5} & \cellcolor{red!50}\bf{5} \\ \hline 
alg.3 & 91.56 & 94.11 & 2.56 & 1.43 & 1739 & \cellcolor{red!20}\bf{4} & \cellcolor{red!20}\bf{4} & \cellcolor{red!20}\bf{4} \\ \hline 
alg.4 & 93.44 & 93.78 & 0.33 & 1.05 & 88612 & \cellcolor{green!20}\bf{2} & \cellcolor{green!50}\bf{1} & \cellcolor{green!50}\bf{1} \\ \hline 
alg.5 & 98.00 & 97.67 & 0.33 & 1.17 & 33266 & \cellcolor{green!50}\bf{1} & \cellcolor{yellow!35}{\bf 3} & \cellcolor{green!20}\bf{2} \\ 
\end{tabular}
\end{center} 
\vskip 12pt
\begin{center}
\begin{tabular}{c||c|c||c|c|c||c|c|c} 
\rule{0pt}{1em} \cellcolor{black!18}{\bf Indian} & IF & IM & $D_{err}$ & $R_{err}$ & $N$ & $D_{err}\downarrow$ & $R_{err}\downarrow$\ & \ \ $N\uparrow$\ \ \\ \hline
alg.1 & 86.78 & 90.78 & 4.00 & 1.43 & 1058 & \cellcolor{red!50}\bf{5} & \cellcolor{red!20}\bf{4} & \cellcolor{red!50}\bf{5} \\ \hline 
alg.2 & 96.00 & 96.11 & 0.11 & 1.03 & 536364 & \cellcolor{green!50}\bf{1} & \cellcolor{green!50}\bf{1} & \cellcolor{green!50}\bf{1} \\ \hline 
alg.3 & 94.56 & 95.56 & 1.00 & 1.23 & 8023 & \cellcolor{green!20}\bf{2} & \cellcolor{green!20}\bf{2} & \cellcolor{green!20}\bf{2} \\ \hline 
alg.4 & 94.89 & 93.67 & 1.22 & 1.24 & 6189 & \cellcolor{yellow!35}{\bf 3} & \cellcolor{yellow!35}{\bf 3} & \cellcolor{yellow!35}{\bf 3} \\ \hline 
alg.5 & 98.56 & 97.22 & 1.33 & 1.93 & 1920 & \cellcolor{red!20}\bf{4} & \cellcolor{red!50}\bf{5} & \cellcolor{red!20}\bf{4} \\ 
\end{tabular}
\end{center} 
\vskip 12pt
\begin{center}
\begin{tabular}{c||c|c||c|c|c||c|c|c} 
\rule{0pt}{1em} \cellcolor{black!18}{\bf White} & WF & WM & $D_{err}$ & $R_{err}$ & $N$ & $D_{err}\downarrow$ & $R_{err}\downarrow$\ & \ \ $N\uparrow$\ \ \\ \hline
alg.1 & 76.67 & 86.22 & 9.56 & 1.69 & 279 & \cellcolor{red!50}\bf{5} & \cellcolor{green!50}\bf{1} & \cellcolor{yellow!35}{\bf 3} \\ \hline 
alg.2 & 89.44 & 96.56 & 7.11 & 3.07 & 205 & \cellcolor{red!20}\bf{4} & \cellcolor{yellow!35}{\bf 3} & \cellcolor{red!20}\bf{4} \\ \hline 
alg.3 & 90.11 & 97.11 & 7.00 & 3.42 & 193 & \cellcolor{yellow!35}{\bf 3} & \cellcolor{red!20}\bf{4} & \cellcolor{red!50}\bf{5} \\ \hline 
alg.4 & 92.67 & 95.78 & 3.11 & 1.74 & 946 & \cellcolor{green!20}\bf{2} & \cellcolor{green!20}\bf{2} & \cellcolor{green!50}\bf{1} \\ \hline 
alg.5 & 95.78 & 98.78 & 3.00 & 3.46 & 462 & \cellcolor{green!50}\bf{1} & \cellcolor{red!50}\bf{5} & \cellcolor{green!20}\bf{2} \\ 
\end{tabular}
\end{center} 
\caption{The TPRs of five algorithms on four races (Asian, Black, Indian, and White) and two genders (female in column 2, male in column 3). We report $D_{err}$, $R_{err}$, and $N$ (columns 4-6) and rank the algorithms for bias according to these three bias metrics (columns 7-9).}
\label{tbl:realExample}
\end{table*} 


\subsection{Behavior under constant error ratios or differences}
\label{sec:sec4-2}

For a better insight into the differences between the proposed metric and error differences or ratios, we study the behavior of the former as the latter remain constant. First, we have: 

\begin{proposition}
\label{prop:constantDiff}
Let $0\leq \varepsilon_1 \leq \varepsilon_2 \leq 1/2$, and $d\in (0, 1/2 - \varepsilon_2)$. We have
$$
N_{\alpha,\beta} (\varepsilon_1 + d, \varepsilon_2 + d) \leq
N_{\alpha,\beta} (\varepsilon_1, \varepsilon_2). 
$$
\end{proposition}

\begin{proof}
From the formula for $N_{\alpha,\beta}$ in Eq.~\ref{eq:twoproportions-arcSin}, it suffices to show that 
$$
\arcsin(\sqrt{\varepsilon_1})-\arcsin(\sqrt{\varepsilon_2}) \leq  \arcsin(\sqrt{\varepsilon_1 + d})-\arcsin(\sqrt{\varepsilon_2 + d})
$$
or, 
$$
\arcsin(\sqrt{\varepsilon_1}) - \arcsin(\sqrt{\varepsilon_1 + d}) \leq \arcsin(\sqrt{\varepsilon_2}) - \arcsin(\sqrt{\varepsilon_2 + d}) 
$$
or, that 
$$
f(x) = \arcsin(\sqrt{x})-\arcsin(\sqrt{x + d})
$$
is increasing function of $x$ for $0\leq x \leq 1/2$ and $d\in (0, 1/2 - x)$. For that, it suffices to show that the derivative of $f$ with respect to $x$ is positive, that is,
$$
0 < f'(x) = \frac{1}{2 \sqrt{1 - x} \sqrt{x}} - \frac{1}{2 \sqrt{-x - d + 1} \sqrt{x + d}} 
$$
giving, 
$$
\frac{1}{2 \sqrt{-x - d + 1} \sqrt{x + d}} < \frac{1}{2 \sqrt{1 - x} \sqrt{x}} 
$$
or, 
$$
2 \sqrt{1 - x} \sqrt{x} < 2 \sqrt{-x - d + 1} \sqrt{x + d}
$$
or, 
$$
(1-x)x < (1-x-d) (x+d)
$$
which holds, as on both sides the sum of the two factors is 1, while the absolute value of their difference is $1-2x$ on the left hand side, larger than the $1-2x-2d$ on the right hand side. 
\end{proof}

Proposition~\ref{prop:constantDiff} highlights an important difference between the proposed metric and $D_{err}$. As we keep the bias measured by $D_{err}$ constant, the bias measured by $N_{\alpha,\beta}$ increases with the accuracy of the algorithm. Table~\ref{tbl:example-constantDiff} illustrates this property for $c=0.10$. 

\begin{table}[ht]
\begin{center}
\begin{tabular}{c|c|c|c|c||c|c} 
\rule{0pt}{1em} & $\varepsilon_1$ & $\varepsilon_2$ & $D_{err}$ & $N$ & $D_{err}\downarrow$ & \ \ $N\uparrow$\ \ \\ \hline
alg.1 & 0.20 & 0.30 & 0.10 & 319 & \bf{=} & \cellcolor{green!50}\bf{1} \\ \hline 
alg.2 & 0.10 & 0.20 & 0.10 & 213 & \bf{=} & \cellcolor{yellow!35}\bf{2} \\ \hline 
alg.3 & 0.00 & 0.10 & 0.10 & 42 & \bf{=} & \cellcolor{red!50}\bf{3} \\ \hline 
\end{tabular} 
\end{center}
\caption{Illustration of Proposition~\ref{prop:constantDiff}. The difference $D_{err} = \varepsilon_2 - \varepsilon_1$ is constant at 0.10. The value of $N$ decreases as $\varepsilon_1, \varepsilon_2$ decrease.}
\label{tbl:example-constantDiff}
\end{table}


Proposition~\ref{prop:constantRatio} is the equivalent of Proposition~\ref{prop:constantDiff}. This time, the error ratios are kept constant, and unlike Proposition~\ref{prop:constantRatio}, the bias as measured by the proposed metric decreases as the algorithms become more accurate. 

\begin{proposition}
\label{prop:constantRatio}

Let $0\leq \varepsilon_1 \leq \varepsilon_2 \leq 1/2$, and $r\in (1, 1 / (2\varepsilon_2))$. We have
$$
N_{\alpha,\beta} (\varepsilon_1, \varepsilon_2) \leq N_{\alpha,\beta} (r\varepsilon_1, r\varepsilon_2).    
$$

\begin{proof}
From the formula for $N_{\alpha,\beta}$ in Eq.~\ref{eq:twoproportions-arcSin}, it suffices to show that 
$$
\arcsin(\sqrt{r\varepsilon_1})-\arcsin(\sqrt{r\varepsilon_2}) \leq
\arcsin(\sqrt{\varepsilon_1})-\arcsin(\sqrt{\varepsilon_2}) 
$$
or, 
$$
\arcsin(\sqrt{\varepsilon_2}) - \arcsin(\sqrt{r\varepsilon_2})
\leq
\arcsin(\sqrt{\varepsilon_1}) - \arcsin(\sqrt{r\varepsilon_1}) 
$$
or, that 
$$
f(x) = \arcsin(\sqrt{x})-\arcsin(\sqrt{rx})
$$
is an decreasing function of $x$ for $0\leq x \leq 1/2$ and $r\in (1, 1/(2x))$. It suffices to show that the derivative of $f$ with respect to $x$ is negative 
$$
f'(x) = \frac{1}{2 \sqrt{1 - x} \sqrt{x}} - \frac{r}{2 \sqrt{rx} \sqrt{1 - rx}} < 0 
$$
or, 
$$
\frac{1}{2 \sqrt{1 - x} \sqrt{x}} < \frac{r}{2 \sqrt{rx} \sqrt{1 - rx}}
$$ 
or, after the simplifications 
$$
\sqrt{1-rx} < \sqrt{r} \sqrt{1-x}
$$
or, 
$$\sqrt{1-rx} < \sqrt{r-rx}$$ 
which holds for $1 < r < 1/2x$. 
\end{proof}
\end{proposition}
Table~\ref{tbl:example-constantRatio} illustrates the property described by the Proposition~\ref{prop:constantRatio} for $r=2$. 
\begin{table}[ht]
\begin{center}
\begin{tabular}{c|c|c|c|c||c|c} 
\rule{0pt}{1em} & $\varepsilon_1$ & $\varepsilon_2$ & $R_{err}$ & $N$ & $R_{err}\downarrow$ & \ \ $N\uparrow$\ \ \\ \hline
alg.1 & 0.20 & 0.40 & 2.00 & 88 & \bf{=} & \cellcolor{red!50}\bf{3} \\ \hline 
alg.2 & 0.10 & 0.20 & 2.00 & 213 & \bf{=} & \cellcolor{yellow!35}\bf{2} \\ \hline 
alg.3 & 0.05 & 0.10 & 2.00 & 463 & \bf{=} & \cellcolor{green!50}\bf{1} \\ \hline 
\end{tabular} 
\end{center}
\caption{Illustration of Proposition~\ref{prop:constantRatio}. The ratio of the erros $R_{err} = \varepsilon_2 / \varepsilon_1$ is constant at 2. The value of $N$ increases as $\varepsilon_1, \varepsilon_2$ decrease.}
\label{tbl:example-constantRatio}
\end{table}


\subsection{Bias metric behavior for diminishing error values}

As algorithms continuously improve, the evolution of the state-of-the-art in a classification problem corresponds to a sequence of increasingly accurate algorithms, with errors that, ideally, would tend to zero on both demographics:
$$
(\varepsilon_1, \varepsilon_2) \rightarrow (0,0).
$$ 
For an insight into the behavior of the three metrics near $(0,0)$, Figure~\ref{fig:contours} shows their contours as functions of $\varepsilon_1, \varepsilon_2$. We notice that the error ratio suffers from a singularity in the vicinity of $(0,0)$, meaning that it is an unstable bias metric for highly accurate algorithms. Indeed, its value depends on the direction from which $(0,0)$ is approached and in fact, it is constant along each such direction.  
\begin{figure*}
\begin{center}
        \includegraphics[width=0.33\textwidth]{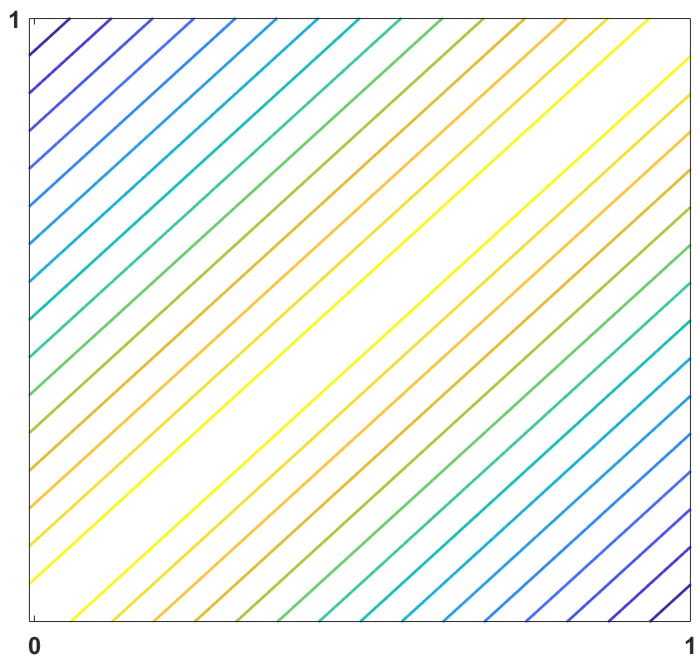} \hfill
        \includegraphics[width=0.33\textwidth]{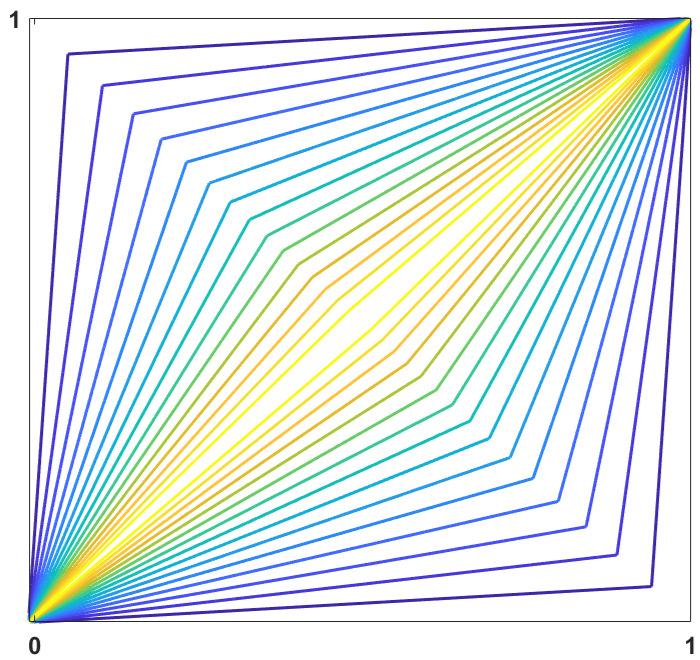} \hfill
        \includegraphics[width=0.33\textwidth]{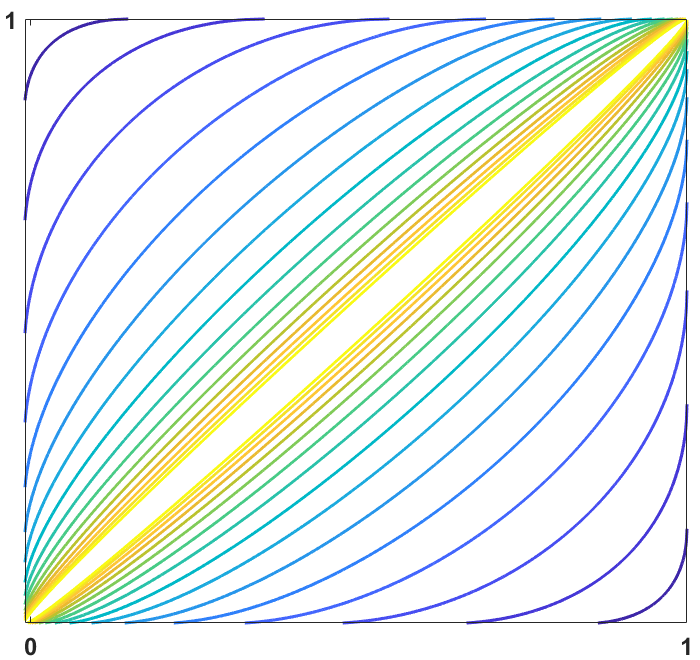} 
\end{center}
\caption{From left to right: the contours of the functions $D_{err}(\varepsilon_1, \varepsilon_2) = \varepsilon_2 - \varepsilon_1$, 
$R_{err}(\varepsilon_1, \varepsilon_2) = \varepsilon_2 / \varepsilon_1$, and $N_{\alpha,\beta} (\varepsilon_1,\varepsilon_2)$ in Eq.~\ref{eq:twoproportions-arcSin}.}
\label{fig:contours}
\end{figure*}

Proposition~\ref{prop:proposedLimitZero} shows that the sample size metric does not suffer from a singularity at $(0,0)$. That is, similarly to the error differences, and unlike the error ratios, the proposed metric does have a limit at $(0,0)$. 

\begin{proposition}
\label{prop:proposedLimitZero}
Let $0 \leq \varepsilon_1 \leq \varepsilon_2$. We have 
$$
\lim_{\varepsilon_1, \varepsilon_2 \rightarrow 0^+} N_{\alpha,\beta} (\varepsilon_1, \varepsilon_2) = \infty
$$
\end{proposition}

\begin{proof}
Since 
$$
{\lim_{x \rightarrow 0}} \arcsin(x) = 0
$$ 
we have 
$$
\lim_{\varepsilon_1, \varepsilon_2 \rightarrow 0^+} (\arcsin(\varepsilon_1) - (\arcsin(\varepsilon_2)) = 0
$$ 
giving 
$$
\lim_{\varepsilon_1, \varepsilon_2 \rightarrow 0^+} N_{\alpha,\beta} (\varepsilon_1, \varepsilon_2) = \infty
$$
\end{proof}
This is an intuitively expected result. Indeed, any reasonable statistical test would require a number of negative samples to be able to distinguish performance differences over the two demographics. Thus, as the error rates on both demographics decrease, we need a larger number of samples to obtain the required negative samples with a sufficiently high probability. 

In the next section, we discuss properties of the proposed bias metric that can be deduced from such types of arguments. Apart from the generality of that alternative approach, which is independent of the statistical test we use, we note that it is not reliant to an approximate formula for the sample size, as for example Eq.~\ref{eq:twoproportions-arcSin}, and thus, there are no validity issues related to the precision of the approximation when extremal properties are studied. 

%% file: sec5.tex
\section{Properties independent of the choice of statistical test} 
\label{sec:sec5}

Although there is no generally accepted notion of a {\em proper} statistical test, there are some desirable properties that one would expect to be common to all reasonable tests, under the minimal assumption that their outcomes are computed by statistical analysis of the input data, rather than being arbitrary. 

One such property is that a test should detect no bias if, with high probability, there are no misclassifications in a sample of size $N$. We formalize this property in a setting of more than two demographics, generalizing Proposition~\ref{prop:proposedLimitZero}. 

Consider a test set consisting of $k$ demographics, $D_1, D_2, \dots, D_k$, and let the respective error rates of a classification algorithm be $\varepsilon_1, \varepsilon_2, \dots, \varepsilon_k$. According the general principle stated in Section~\ref{sec:sec1}, the bias of the algorithm is inversely commensurable to the sample size required by a statistical test to determine whether there are differences in the error rates among demographics. Thus, the null hypothesis is 
\begin{equation}
\label{eq:nullGeneral}
H_0 : \varepsilon_1 = \varepsilon_2 = \cdots = \varepsilon_k
\end{equation}
and we have the following. 
\begin{proposition}
\label{prop:generalisation3-3}
Let $N_{\alpha,\beta}(\varepsilon_1, \varepsilon_2, \cdots, \varepsilon_k)$ be the required sample size for a proper statistical test of $H_0$, with confidence $\alpha$ and power $1-\beta$. We have 
\begin{equation}
\lim_{\varepsilon_1, \varepsilon_2, \dots, \varepsilon_k \rightarrow 0^+} N_{\alpha,\beta}(\varepsilon_1, \varepsilon_2, \cdots, \varepsilon_k) = \infty 
\label{eq:prop5-1}
\end{equation}
\end{proposition}

\begin{proof}
Let $\varepsilon_1, \varepsilon_2, \dots, \varepsilon_k < 1/n^2$, with $n\in \mathbb{N}$. The probability that a sample of size $n$ contains no misclassifications is at least 
$$
(\frac{n^2-1}{n^2})^n \rightarrow 1 \quad \text{as} \quad n\rightarrow \infty
$$
and thus, with high probability, a sample of size $n$ contains no misclassifications as 
${\varepsilon_1, \varepsilon_2, \dots, \varepsilon_k} \rightarrow 0^+$. From the desirable property of the proper statistical tests that they should detect no bias when, with high probability, the sample contains no misclassifications, we have Eq.~\ref{eq:prop5-1}. 
\end{proof}

Proposition~\ref{prop:generalisation3-3} is a test agnostic generalization of Proposition~\ref{prop:proposedLimitZero}. While it is out of the scope of the paper to formally define a complete set of desirable properties of proper statistical sets, we note that Propositions~\ref{prop:noBias},~\ref{prop:biasMonotonicity}, also correspond to properties that one expects to hold for any reasonable statistical test. 

For example, Proposition~\ref{prop:noBias} can be seen an instance of the more general intuitive property that when the error rates tend to be equal, a test would require larger sample sizes to tell them apart. That is, for the null hypothesis in Eq.~\ref{eq:nullGeneral}, we should have 
$$
\lim_{\varepsilon_1, \varepsilon_2, \dots, \varepsilon_{k-1} \rightarrow \varepsilon_k} N_{\alpha,\beta} (\varepsilon_1, \varepsilon_2, \dots, \varepsilon_k)= \infty 
$$

Similarly, Proposition~\ref{prop:biasMonotonicity} can be seen as instance of a more general property, stating that when the largest error rate increases, while all the other errors stay the same, the bias also increases. That is, if 
$$
0\leq \varepsilon_1 \leq \varepsilon_2 \leq \dots \leq \varepsilon_{k-1} \leq \varepsilon_k \leq \varepsilon_{k+1} \leq 1/2
$$
we should have
$$
N_{\alpha,\beta}(\varepsilon_1, \varepsilon_2, \dots, \varepsilon_{k-1},\varepsilon_{k+1}) \leq N_{\alpha,\beta}(\varepsilon_1, \varepsilon_2, \dots, \varepsilon_{k-1},\varepsilon_k).
$$

%% file: sec6.tex
\section{Discussion} 
\label{sec:sec6}

Here we discuss some aspects of the proposed metric which were not covered in the previous sections, as well as some potential alternative approaches, the detailed development of which is out of the scope of this paper. 

Sample size computations, either analytical, as for example in Eq.~\ref{eq:twoproportions-arcSin}, or through Monte Carlo simulations, are based on assumptions on the distributions of the errors, for example, the $\chi$-squared distribution for the $\chi$-squared test. Thus, bias measures could be developed with direct references to such distributions. However, the sample size approach utilizes a wealth of accumulated theoretical work and practical experience on which distributions are the most appropriate and how computations on them can be performed efficiently. Still, working directly with error distributions remains a possibility in white-box approaches, where experimental evidence on the error distributions in a specific problem domain may be incorporated in the computation of sample sizes. 

The use of a Bayesian framework is also a possibility for an equivalent approach to the measurement of bias. However, the Bayesian approach seems more natural when the error rates and their distributions are not fixed, but updated as samples are processed during the validation of the algorithm. We reiterate here that our aim is not to propose a method for computing distributions of error rates, but to establish a measure of bias from point estimates of the error rates on the various demographics. 

Connections between the proposed approach and information-theory could also be worth investigating. For example, as indicated in the discussion of Propositions~\ref{prop:proposedLimitZero},~\ref{prop:generalisation3-3}, the larger sample size required for bias detection in highly accurate algorithms could be attributed to the low entropy of random variables over the demographics. That is, if an algorithm is highly accurate overall, then, on all demographics, each classification outcome gives a small only amount of information on that algorithm's biases. 

Throughout the paper, there was a tacit assumption of equal size samples from the two demographics. Other approaches, such as considering sample sizes proportional to the sizes of the demographics, could be interesting but would also introduce complexities which, we believe, would be best addressed within the specific contexts of practical applications.

\subsection{Human perception issues}
\label{sec:sec6-1}

As we increasingly rely on algorithms for ever more impactful decisions, paraphrasing a well-known legal dictum \cite{hewart1924}, it does not suffice for an algorithm to be fair, it must also be seen to be fair. Thus, the perception of bias is an aspect that should also be studied, and here we postulate the presence of Weber-type laws. 

Figure~\ref{fig:perception} illustrates an instance of Weber's law on the human {\em sense of number} and our judgment of numerosity \cite{dehaene2003}. Although in visual settings, as the one in Figure~\ref{fig:perception}, Weber's law is considered well-established, it has also been suggested that the empirical findings of the relevant studies depend on the circumstances of the task and stimuli \cite{testolin2021}. Thus, it seems an interesting direction for future work to design experiments aiming at verifying Weber's law in settings, tasks, and stimuli resembling human interactions with deployed machine learning classifiers. 

A further complication that such an empirical study of human perception would need to address is that people might focus disproportionately on the high-entropy outcomes of the experiment, that is, the misclassifications. However, in practice, that could be justified by the higher costs often associated with such outcomes. In the example in Figure~\ref{fig:perception} (top), that means that the attention focuses on the black dots, while the white dots are essentially perceived as background. 

Nevertheless, as an initial observation, we note that the behavior of $N$ described in Proposition~\ref{prop:constantDiff} seems to align well with the human perception of bias, or at least better than the difference $\varepsilon_2 - \varepsilon_1$. That is, as Figure~\ref{fig:perception} suggests, we expect to be easier to discern differences in algorithmic performance over demographics when the number of misclassifications increases while the difference of their absolute numbers stays constant. 

\begin{figure}
    \centering
    \includegraphics[width=0.49\columnwidth]{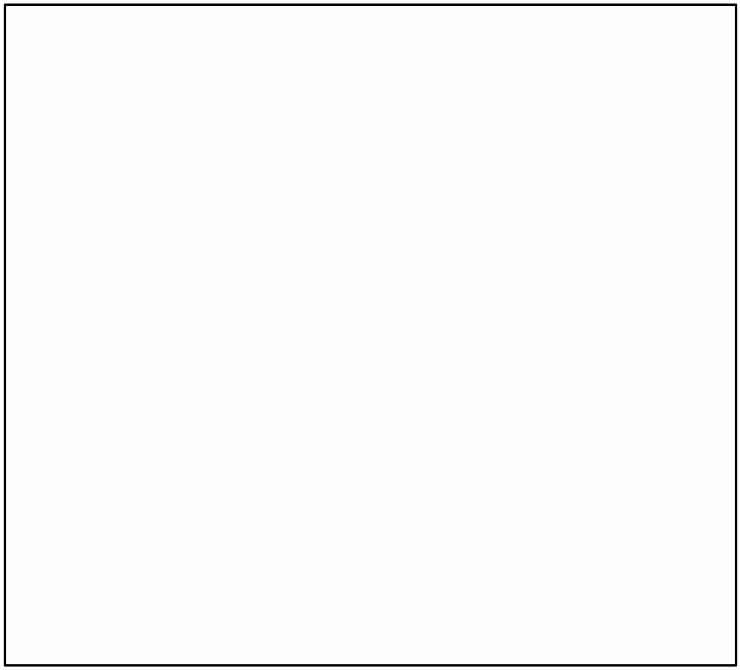}\hfill 
    \includegraphics[width=0.49\columnwidth]{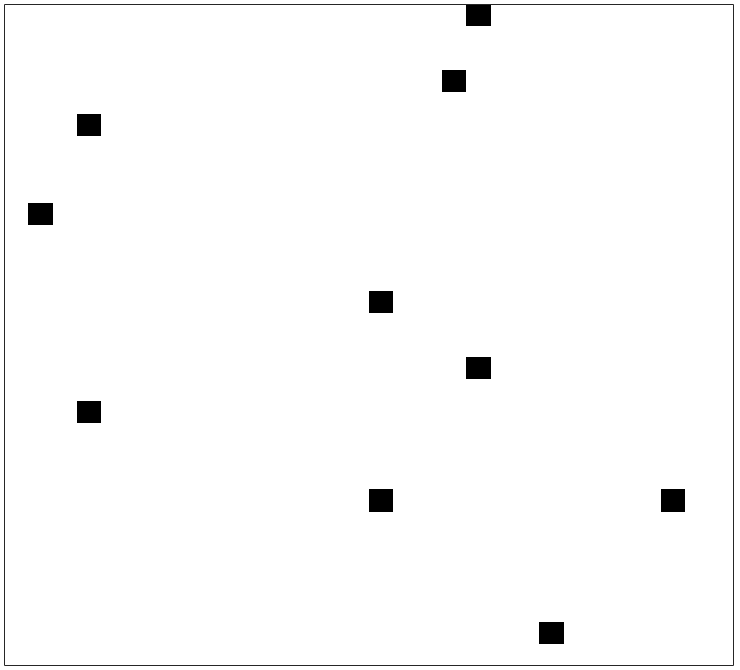}\hfill 
    \includegraphics[width=0.49\columnwidth]{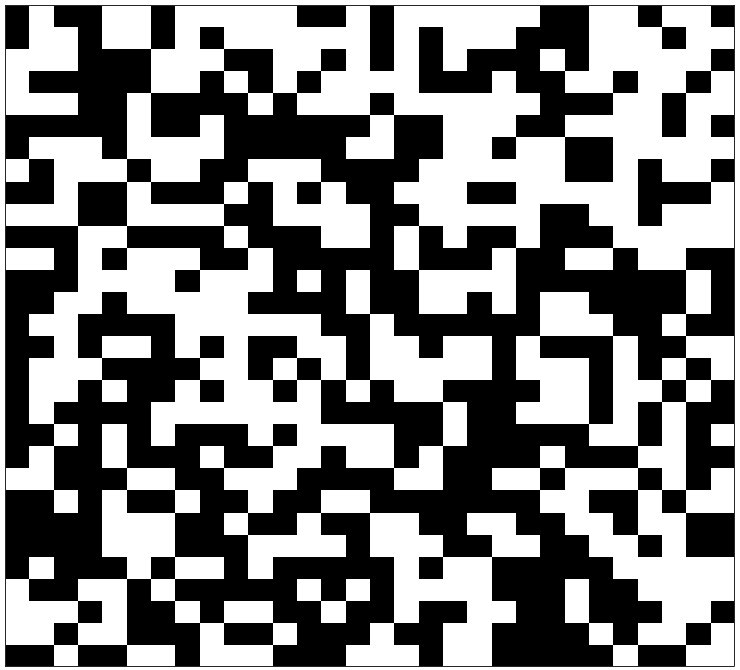}\hfill 
    \includegraphics[width=0.49\columnwidth]{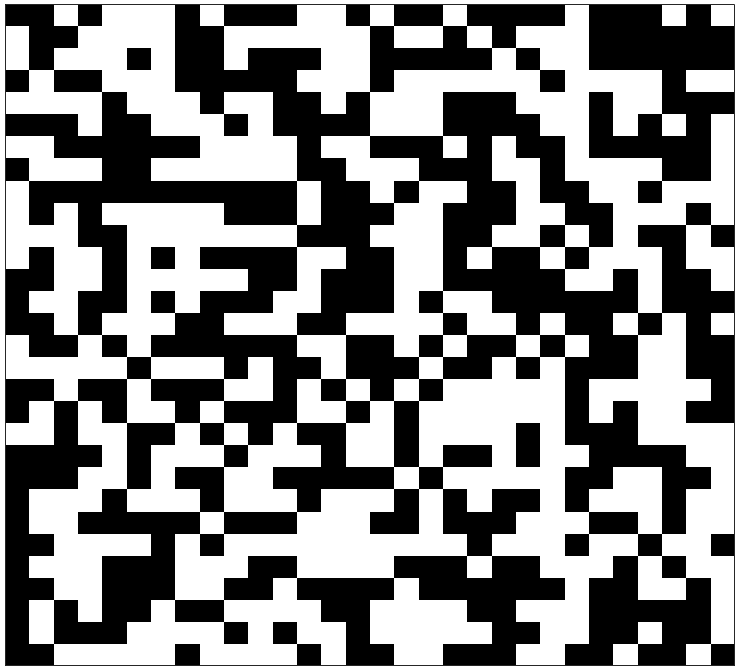}\hfill 
    \caption{In each image, in a $30 \times 30$ grid, a black dot corresponds to a misclassified sample, and a white dot to a correctly classified sample. {\bf Top:} $\varepsilon = 0.00$ (left) against $\varepsilon = 0.01$ (right). {\bf Bottom:} $\varepsilon = 0.49$ (left) against $\varepsilon = 0.50$ (right).}
    \label{fig:perception}
\end{figure}


Finally, in an issue related to dissemination, the insights of this paper could be used to inform the practices of communicating research findings in layman terms. In a typical example from medical research, where classification would correspond to a diagnosis, layman summaries sometimes report just the ratio of positive diagnoses over two groups $r = \varepsilon_2 / \varepsilon_1$, without reporting any of the $\varepsilon_1, \varepsilon_2$, even though that could help assess the implications of the findings. Interestingly, although the size of the study $N$ is often included in layman summaries, information needed for its interpretation might again be missing. Large values of $N$ in particular, might correspond to high values of power and confidence in the design, that is, high values of $\alpha$ and $1 - \beta$, or to small values of $\varepsilon_1, \varepsilon_2$. 

%% file: sec7.tex
\section{Conclusion} 
\label{sec:sec7} 

We proposed a measure of algorithmic bias, understood here as differences in the algorithm's performance over various demographics. Our approach is based on a simple intuitive principle: when the error rates are very different, small statistical tests suffice to establish that fact; if, in contrast, error rates are quite similar, larger tests are needed. We studied the simplest case of two demographics, using a closed-form formula for the sample size of the $\chi$-squared test, establishing various desirable properties of the proposed approach. We then compared it against two commonly used metrics, the differences and the ratios of the error rates, and demonstrated that it produces essentially different bias rankings. Finally we discuss how some of the desirable properties of the proposed metric arise naturally in sample size analysis and should be considered test-independent. 

Although, as we have shown in this paper, the bias ranking results are independent of the values of the parameters $\alpha, \beta$ of the $\chi$-squared test, one expects that a different statistical test, with different assumptions for the error distributions would result to different rankings. In the future, we plan to explore this research avenue, aiming at obtaining further interesting insights into the notion of algorithmic bias. 